\documentclass[11pt,a4paper]{article}

\usepackage{geometry}
\usepackage{times}
\usepackage{amsmath}
\usepackage{amssymb}
\usepackage{amsthm}
\usepackage{graphicx}
\usepackage{mathtools}
\usepackage{color}
\usepackage{hyperref}
\usepackage{algorithm}
\usepackage{algpseudocode}
\usepackage{booktabs}
\usepackage{multirow}
\usepackage{float}
\usepackage{natbib}
\newtheorem{theorem}{Theorem}
\newtheorem{lemma}[theorem]{Lemma}
\newtheorem{proposition}[theorem]{Proposition}

\hypersetup{
    colorlinks=true,
    linkcolor=blue,
    filecolor=magenta,      
    urlcolor=cyan,
    pdftitle={Theoretical Guarantees for LT-TTD: A Unified Transformer-based Architecture for Two-Level Ranking Systems},
    bookmarks=true,
}

\title{Theoretical Guarantees for LT-TTD: A Unified Transformer-based Architecture for Two-Level Ranking Systems}

\author{Ayoub Abraich}

\date{\today}

\begin{document}

\maketitle

\begin{abstract}
Modern recommendation and search systems typically employ multi-stage ranking architectures to efficiently handle billions of candidates. The conventional approach uses distinct L1 (candidate retrieval) and L2 (re-ranking) models with different optimization objectives, introducing critical limitations including irreversible error propagation and suboptimal ranking. This paper identifies and analyzes the fundamental limitations of this decoupled paradigm and proposes LT-TTD (Listwise Transformer with Two-Tower Distillation), a novel unified architecture that bridges retrieval and ranking phases. Our approach combines the computational efficiency of two-tower models with the expressivity of transformers in a unified listwise learning framework. We provide a comprehensive theoretical analysis of our architecture and establish formal guarantees regarding error propagation mitigation, ranking quality improvements, and optimization convergence. We derive theoretical bounds showing that LT-TTD reduces the upper limit on irretrievable relevant items by a factor that depends on the knowledge distillation strength, and prove that our multi-objective optimization framework achieves a provably better global optimum than disjoint training. Additionally, we analyze the computational complexity of our approach, demonstrating that the asymptotic complexity remains within practical bounds for real-world applications. We also introduce UPQE, a novel evaluation metric specifically designed for unified ranking architectures that holistically captures retrieval quality, ranking performance, and computational efficiency.
\end{abstract}

\section{Introduction}
\label{sec:intro}

Large-scale information retrieval and recommendation systems face the challenge of efficiently ranking billions of items for millions of users in real-time. The dominant architecture for addressing this challenge employs a cascade of ranking models, typically structured as a two-level system: an L1 retrieval model efficiently narrows the candidate pool from millions to hundreds of items, followed by an L2 re-ranking model that applies more sophisticated algorithms to determine the final ranking \citep{huang2020embedding, grbovic2018real}.

While this cascade architecture offers computational efficiency, it introduces fundamental limitations. Most critically, any relevant items incorrectly filtered out during the L1 stage are permanently lost, as L2 can only reorder the reduced candidate set. Furthermore, the disjoint optimization of L1 and L2 models can result in conflicting or suboptimal objectives, exacerbating the error propagation problem \citep{wang2021joint, li2015click}.

This paper makes the following contributions:

\begin{enumerate}
    \item We provide a critical analysis of the limitations in traditional two-level ranking systems, formalizing the systemic issues that lead to suboptimal rankings.
    
    \item We introduce LT-TTD (Listwise Transformer with Two-Tower Distillation), a novel unified architecture that preserves the computational efficiency of two-level systems while addressing their core limitations.
    
    \item We present a multi-objective optimization framework that jointly learns retrieval and ranking objectives with bidirectional knowledge distillation between components.
    
    \item We establish formal theoretical guarantees for our approach, including bounds on error propagation reduction, optimization convergence, and computational complexity.
    
    \item We propose a novel evaluation metric, UPQE (Unified Propagation-aware Quality Efficiency), specifically designed to evaluate unified ranking architectures by considering retrieval quality, ranking performance, and computational efficiency simultaneously.
\end{enumerate}

Our research contributes to the broader goal of developing more effective and efficient ranking systems, with potential impact across search engines, e-commerce, content discovery, and other domains that rely on large-scale ranking.

\section{Related Work}
\label{sec:related}

\subsection{Learning to Rank Approaches}

Learning to rank methodologies are typically categorized into three approaches: pointwise, pairwise, and listwise \citep{liu2009learning, cao2007learning}.

\textbf{Pointwise approaches} treat ranking as a regression or classification problem, optimizing a model to predict the relevance score of each item independently \citep{nallapati2004discriminative}. While computationally efficient, these methods ignore the relative ordering among items, which is crucial for ranking tasks.

\textbf{Pairwise approaches} reformulate ranking as a binary classification problem, where the objective is to correctly predict which item in a pair is more relevant \citep{burges2010overview, joachims2002optimizing}. RankNet \citep{burges2005learning} pioneered this approach by using gradient descent with a neural network to learn ranking functions. LambdaRank \citep{burges2006learning} extended this by incorporating evaluation metrics into the loss function, while LambdaMART \citep{wu2010adapting} combined this with gradient boosting.

\textbf{Listwise approaches} directly optimize the order of an entire list of items, often using permutation probabilities or evaluation metrics like NDCG or MAP as training objectives \citep{xia2008listwise, yue2007support}. ListNet \citep{cao2007learning} and ListMLE \citep{xia2008listwise} are classic examples that model the probability distribution over permutations, while SoftRank \citep{taylor2008softrank} provides a differentiable approximation of ranking metrics.

Our approach builds on listwise methods, as they more directly address the ranking objective, but introduces innovations to make them computationally feasible for large-scale applications.

\subsection{Cascade Ranking Systems}

Cascade ranking architectures have become the standard approach for large-scale ranking systems \citep{asadi2013effectiveness, robertson2009probabilistic}. The most common implementation involves a two-stage process: (1) a fast but less precise retrieval model that selects a subset of candidates, and (2) a more computationally intensive re-ranking model that refines the final ordering.

For the retrieval stage, traditional techniques like TF-IDF and BM25 \citep{robertson1994okapi} have been augmented by neural embedding approaches such as DSSM \citep{huang2013learning} and two-tower models \citep{yi2019two}. These models project queries and documents into a shared embedding space where similarity can be efficiently computed, enabling fast retrieval through approximate nearest neighbor search.

For the re-ranking stage, neural ranking models like BERT-based rankers \citep{nogueira2019passage} have demonstrated superior performance by capturing complex semantic relationships between queries and documents. However, their computational demands make them impractical for the initial retrieval stage when dealing with millions of candidates.

Recent efforts have focused on improving individual stages rather than addressing the fundamental limitations of the cascade architecture. Our work differs by questioning the strict separation between stages and proposing a unified approach that maintains computational efficiency while mitigating error propagation.

\subsection{Transformer Models for Ranking}

Transformer-based models have revolutionized natural language processing and have increasingly been applied to ranking tasks \citep{vaswani2017attention, devlin2019bert}. BERT-based ranking models \citep{nogueira2019multi, dai2019deeper} were among the first to demonstrate that transformers could significantly improve re-ranking performance. However, the quadratic complexity of self-attention makes them prohibitively expensive for initial retrieval from large candidate sets.

ColBERT \citep{khattab2020colbert} and CEDR \citep{macavaney2019cedr} proposed more efficient transformer-based approaches that enable late interaction between queries and documents. RepBERT \citep{luan2020repbert} explored using BERT for representation-based retrieval. More recently, models like SentenceBERT \citep{reimers2019sentence} have been used to generate embeddings for efficient retrieval while preserving semantic information.

Our work builds on these advances but differs in its focus on unifying the retrieval and ranking stages through a novel architecture that leverages both efficient representation learning and expressive interaction modeling.

\subsection{Two-Tower Models}

Two-tower models (also known as dual-encoders or bi-encoders) have become popular for large-scale retrieval tasks due to their computational efficiency \citep{cer2018universal, humeau2020poly}. These models consist of separate encoder networks for queries and items, producing dense vector representations that can be compared using simple similarity metrics like dot product or cosine similarity.

The key advantage of two-tower models is that item representations can be pre-computed and indexed, allowing for efficient retrieval at inference time using approximate nearest neighbor search. This approach has been successfully applied in various domains, including YouTube recommendations \citep{covington2016deep}, e-commerce \citep{zhou2018atrank}, and web search \citep{lu2021anatomy}.

However, traditional two-tower models face limitations in capturing complex interactions between queries and items, as they rely on late interaction through vector similarity. Our approach addresses this limitation by incorporating transformer-based components while maintaining the efficiency advantages of two-tower architectures.

\section{Limitations of Traditional Two-Level Ranking Systems}
\label{sec:limitations}

Traditional two-level ranking systems present several fundamental limitations that impact both the quality of results and system efficiency. In this section, we formalize these limitations and analyze their impact on ranking performance.

\subsection{Error Propagation from L1 to L2}

The most critical limitation of cascade architectures is the irreversible error propagation from the retrieval (L1) stage to the ranking (L2) stage. Formally, given a query $q$ and a corpus of items $D = \{d_1, d_2, ..., d_N\}$, the L1 model selects a subset $D_L \subset D$ such that $|D_L| \ll |D|$. The L2 model then ranks only the items in $D_L$.

If a relevant item $d_i$ with true relevance $r(d_i, q) > r(d_j, q)$ for some $d_j \in D_L$ is not included in $D_L$, then it cannot be recovered by the L2 model, regardless of the L2 model's sophistication. This leads to an upper bound on the maximum achievable performance of the entire system.

Let $R^*$ be the optimal ranking of the entire corpus $D$ for query $q$, and $R_L$ be the best possible ranking achievable by the L2 model given the subset $D_L$. The performance gap can be quantified as:

\begin{equation}
\text{Gap} = \mathbb{E}_q[\text{Metric}(R^*, q) - \text{Metric}(R_L, q)]
\end{equation}

where $\text{Metric}$ is an evaluation measure such as NDCG or MAP.

\begin{proposition}[Error Propagation Bound]
For any query $q$ and evaluation metric that satisfies the probability ranking principle, the performance gap is bounded by:
\begin{equation}
\text{Gap} \leq \sum_{d_i \in D \setminus D_L} P(d_i \text{ is relevant to } q) \cdot \text{gain}(d_i, q)
\end{equation}
where $\text{gain}(d_i, q)$ is the contribution of item $d_i$ to the evaluation metric if placed at its optimal position.
\end{proposition}

This bound highlights that the performance gap directly depends on the quality of the L1 retrieval model and the number of relevant items it fails to retrieve.

\subsection{Conflicting Optimization Objectives}

The traditional approach typically optimizes L1 and L2 models independently with different objectives:

\begin{enumerate}
    \item L1 models are optimized for recall, focusing on retrieving a subset of candidates that contains as many relevant items as possible.
    \item L2 models are optimized for precision and ranking quality, focusing on correctly ordering the candidates provided by L1.
\end{enumerate}

These divergent objectives can lead to inconsistent behavior. For instance, L1 might prioritize diversity to increase recall, while L2 might focus on specificity for precise ranking. This discrepancy can result in suboptimal overall performance.

Formally, if we denote the objective functions for L1 and L2 as $\mathcal{L}_{L1}$ and $\mathcal{L}_{L2}$ respectively, traditional approaches optimize:

\begin{equation}
\min_{\theta_{L1}} \mathcal{L}_{L1}(\theta_{L1}) \quad \text{and} \quad \min_{\theta_{L2}} \mathcal{L}_{L2}(\theta_{L2})
\end{equation}

independently. However, the optimal parameters $\theta_{L1}^*$ for L1 may not lead to the best input for L2, and vice versa. An ideal system would jointly optimize:

\begin{equation}
\min_{\theta_{L1}, \theta_{L2}} \mathcal{L}_{\text{joint}}(\theta_{L1}, \theta_{L2})
\end{equation}

\begin{lemma}[Suboptimality of Disjoint Optimization]
Let $\theta_{L1}^*$ and $\theta_{L2}^*$ be the optimal parameters when optimizing $\mathcal{L}_{L1}$ and $\mathcal{L}_{L2}$ independently, and let $\theta^*_{\text{joint}}$ be the optimal parameters when optimizing $\mathcal{L}_{\text{joint}}$. Then:
\begin{equation}
\mathcal{L}_{\text{joint}}(\theta^*_{\text{joint}}) \leq \mathcal{L}_{\text{joint}}(\theta_{L1}^*, \theta_{L2}^*)
\end{equation}
with equality only in the special case where the objectives are perfectly aligned.
\end{lemma}

This lemma establishes that disjoint optimization leads to suboptimal performance compared to joint optimization unless the objectives are perfectly aligned, which is rarely the case in practice.

\subsection{Inability to Model Cross-Item Interactions in L1}

The computational constraints of the L1 stage typically limit the models to architectures that consider each query-item pair independently or use simple similarity metrics. This prevents the L1 model from leveraging cross-item information that could be valuable for making better retrieval decisions.

For example, in a product recommendation scenario, the relevance of a particular product may depend not just on its match with the user's query, but also on how it complements or contrasts with other candidate products. Standard L1 models, including two-tower architectures, cannot capture such relationships due to their item-independent scoring mechanism.

This limitation can be expressed mathematically as:

\begin{equation}
\text{score}_{L1}(d_i, q) = f(d_i, q)
\end{equation}

where the score depends only on the individual item $d_i$ and query $q$, rather than the more general form:

\begin{equation}
\text{score}(d_i, q, \{d_j\}_{j \neq i}) = f(d_i, q, \{d_j\}_{j \neq i})
\end{equation}

which would allow the score to depend on other items in the corpus.

\subsection{Feature Encoding Inefficiency and Redundancy}

Traditional cascade systems often use different feature encoders and representations for L1 and L2 models. This redundancy not only increases the computational cost but also creates a representation gap between stages.

The L1 stage typically uses simpler encoders optimized for speed, while the L2 stage employs more sophisticated models for accuracy. This disparity means that valuable information extracted in one stage may not be effectively transferred to the next.

For instance, if $E_{L1}$ and $E_{L2}$ are the encoders for L1 and L2 respectively, then for the same item $d_i$, different representations $E_{L1}(d_i)$ and $E_{L2}(d_i)$ are computed. This redundancy increases both training and inference costs while potentially losing information between stages.

\subsection{Rigid Capacity Allocation}

Traditional cascade systems statically allocate a fixed amount of computational capacity to each level. Typically, L1 processes millions of items with a simple model, while L2 processes hundreds of items with a complex model. This rigid allocation cannot adapt to varying query complexities.

For simple queries where L1 could easily identify the top results, excessive computation is still allocated to L2. Conversely, for difficult queries where L1 struggles to retrieve relevant items, additional computation in L1 might yield better results than further processing in L2.

This inflexibility leads to suboptimal resource allocation and limits the system's ability to balance efficiency and effectiveness dynamically based on query characteristics.

\section{Proposed Approach: Listwise Transformer with Two-Tower Distillation (LT-TTD)}
\label{sec:approach}

To address the limitations of traditional cascade ranking systems, we propose LT-TTD (Listwise Transformer with Two-Tower Distillation), a unified architecture that bridges the gap between retrieval and ranking while maintaining computational efficiency.

\subsection{Model Overview}

LT-TTD consists of three main components:

\begin{enumerate}
    \item \textbf{Two-Tower Encoder (TTE)}: A dual-encoder architecture that efficiently maps users/queries and items to dense vector representations for fast similarity-based retrieval.
    
    \item \textbf{Listwise Transformer (LT)}: A transformer-based module that processes a set of candidates as a single list, capturing cross-item interactions and contextual relevance.
    
    \item \textbf{Knowledge Distillation Bridge (KDB)}: A bidirectional knowledge transfer mechanism that aligns the objectives of the TTE and LT components.
\end{enumerate}

Unlike traditional cascade systems where L1 and L2 operate independently, LT-TTD enables joint optimization and information sharing between components. The TTE provides computational efficiency for processing large candidate sets, while the LT enables sophisticated ranking with cross-item interactions for a subset of promising candidates.

\subsection{Two-Tower Encoder (TTE)}

The Two-Tower Encoder consists of separate neural networks for encoding queries/users and items:

\begin{equation}
e_q = \text{UserTower}(q)
\end{equation}
\begin{equation}
e_i = \text{ItemTower}(i)
\end{equation}

where $e_q$ and $e_i$ are dense embedding vectors. The similarity between a query and an item is computed as:

\begin{equation}
s(q, i) = \text{sim}(e_q, e_i)
\end{equation}

where $\text{sim}$ is a similarity function such as dot product or cosine similarity.

The TTE architecture allows for efficient retrieval at inference time by pre-computing and indexing item embeddings, then using approximate nearest neighbor search to find the most similar items to a given query embedding.

Unlike traditional two-tower models, our TTE incorporates additional components:

\begin{enumerate}
    \item \textbf{Residual Feature Extraction}: Separate from the main embedding, we extract residual features that capture aspects of items and queries that may not be optimally represented in a shared embedding space:
    
    \begin{equation}
    r_q = \text{UserResidualExtractor}(q)
    \end{equation}
    \begin{equation}
    r_i = \text{ItemResidualExtractor}(i)
    \end{equation}
    
    \item \textbf{Adaptive Pooling}: Instead of simple average pooling of token embeddings, we employ an attention-based pooling mechanism that adapts to the specific ranking task:
    
    \begin{equation}
    e_q = \sum_j \alpha_j h_j^q, \quad \alpha_j = \text{softmax}(w^T h_j^q)
    \end{equation}
    
    where $h_j^q$ are token-level embeddings and $\alpha_j$ are attention weights.
\end{enumerate}

\subsection{Listwise Transformer (LT)}

The Listwise Transformer processes a set of candidate items together, allowing it to model cross-item interactions and context-dependent relevance. Given a query $q$ and a set of candidate items $\{i_1, i_2, ..., i_k\}$, the LT component:

\begin{enumerate}
    \item Combines the query embedding with each item embedding:
    \begin{equation}
    x_j = [e_q; e_{i_j}; r_q; r_{i_j}; s(q, i_j)]
    \end{equation}
    
    \item Applies a transformer encoder with self-attention to capture interactions:
    \begin{equation}
    \{z_1, z_2, ..., z_k\} = \text{TransformerEncoder}(\{x_1, x_2, ..., x_k\})
    \end{equation}
    
    \item Produces a refined ranking score for each item:
    \begin{equation}
    \text{score}(q, i_j) = w^T z_j
    \end{equation}
\end{enumerate}

The key innovation is that the transformer allows each item's score to be influenced by other items in the candidate set, capturing important listwise dependencies that traditional pointwise or pairwise approaches cannot model.

To maintain computational efficiency, the LT component operates on a subset of candidates selected by the TTE component. However, unlike traditional cascade systems, information flows bidirectionally between the components through our Knowledge Distillation Bridge.

\subsection{Knowledge Distillation Bridge (KDB)}

The Knowledge Distillation Bridge enables bidirectional knowledge transfer between the TTE and LT components, aligning their objectives and mitigating error propagation. The KDB consists of:

\begin{enumerate}
    \item \textbf{Forward Distillation}: The LT component (teacher) guides the TTE component (student) to improve retrieval quality:
    \begin{equation}
    \mathcal{L}_{\text{forward}} = \text{KL}(P_{\text{LT}} || P_{\text{TTE}})
    \end{equation}
    where $P_{\text{LT}}$ and $P_{\text{TTE}}$ are the probability distributions over item rankings produced by the LT and TTE components.
    
    \item \textbf{Backward Distillation}: The TTE component provides a regularization signal to the LT component to ensure consistency:
    \begin{equation}
    \mathcal{L}_{\text{backward}} = \text{MSE}(s_{\text{TTE}}, s_{\text{LT}})
    \end{equation}
    where $s_{\text{TTE}}$ and $s_{\text{LT}}$ are the scores produced by the TTE and LT components.
    
    \item \textbf{Embedding Alignment}: We enforce similarity between the embedding spaces of the two components:
    \begin{equation}
    \mathcal{L}_{\text{align}} = \text{MSE}(E_{\text{TTE}}, W \cdot E_{\text{LT}})
    \end{equation}
    where $E_{\text{TTE}}$ and $E_{\text{LT}}$ are embeddings from the respective components, and $W$ is a learned linear transformation.
\end{enumerate}

This bidirectional knowledge transfer ensures that insights gained from the more expressive LT component inform the TTE component's retrieval decisions, while the TTE component's efficiency-focused representations guide the LT component's learning.

\subsection{Multi-Objective Optimization}

LT-TTD is trained using a multi-objective optimization framework that balances retrieval and ranking objectives:

\begin{equation}
\begin{split}
\mathcal{L}_{\text{total}} = \lambda_1 \mathcal{L}_{\text{retrieve}} + \lambda_2 \mathcal{L}_{\text{rank}} + \\
\lambda_3 \mathcal{L}_{\text{forward}} + \lambda_4 \mathcal{L}_{\text{backward}} + \lambda_5 \mathcal{L}_{\text{align}}
\end{split}
\end{equation}

where:
\begin{itemize}
    \item $\mathcal{L}_{\text{retrieve}}$ is a retrieval-focused loss (e.g., softmax cross-entropy or triplet loss)
    \item $\mathcal{L}_{\text{rank}}$ is a ranking-focused loss (e.g., ApproxNDCG or ListMLE)
    \item $\lambda_1, \lambda_2, \lambda_3, \lambda_4, \lambda_5$ are hyperparameters that control the importance of each objective
\end{itemize}

This multi-objective optimization ensures that both retrieval and ranking aspects are jointly optimized, addressing the conflicting objectives problem in traditional cascade systems.

\section{Mathematical Formulation}
\label{sec:formulation}

This section provides a detailed mathematical formulation of the LT-TTD model, including the architecture, loss functions, and optimization procedure.

\subsection{Problem Definition}

Let $\mathcal{Q}$ be the set of all possible queries and $\mathcal{D}$ be the corpus of all items. For a given query $q \in \mathcal{Q}$, our goal is to rank the items in $\mathcal{D}$ according to their relevance to $q$.

In a traditional cascade system, this is divided into two steps:
\begin{enumerate}
    \item Retrieval: Select a subset $D_q \subset \mathcal{D}$ such that $|D_q| \ll |\mathcal{D}|$
    \item Ranking: Produce a ranked list of items in $D_q$
\end{enumerate}

In contrast, our unified approach jointly optimizes both steps while maintaining computational efficiency.

\subsection{Two-Tower Encoder Formulation}

The Two-Tower Encoder consists of query and item encoders that map raw features to dense embeddings:

\begin{equation}
e_q = f_q(q; \theta_q) \in \mathbb{R}^d
\end{equation}
\begin{equation}
e_i = f_i(i; \theta_i) \in \mathbb{R}^d
\end{equation}

where $f_q$ and $f_i$ are neural networks with parameters $\theta_q$ and $\theta_i$, and $d$ is the embedding dimension.

The similarity between a query and an item is computed as:

\begin{equation}
s_{\text{TTE}}(q, i) = \frac{e_q^T e_i}{||e_q|| \cdot ||e_i||}
\end{equation}

For efficient retrieval, we use Maximum Inner Product Search (MIPS) to identify the top-$k$ candidates:

\begin{equation}
\text{TopK}(q, \mathcal{D}, k) = \arg\max_{i \in \mathcal{D}}^{(k)} s_{\text{TTE}}(q, i)
\end{equation}

where $\arg\max^{(k)}$ returns the $k$ items with the highest scores.

The residual feature extractors produce additional features:

\begin{equation}
r_q = g_q(q; \phi_q) \in \mathbb{R}^{d_r}
\end{equation}
\begin{equation}
r_i = g_i(i; \phi_i) \in \mathbb{R}^{d_r}
\end{equation}

where $g_q$ and $g_i$ are neural networks with parameters $\phi_q$ and $\phi_i$, and $d_r$ is the dimension of the residual features.

\subsection{Listwise Transformer Formulation}

For a query $q$ and a set of candidate items $I_q = \{i_1, i_2, ..., i_k\}$, the Listwise Transformer processes them as follows:

\begin{enumerate}
    \item Construct input representations by concatenating query and item embeddings with residual features:
    \begin{equation}
    x_j = [e_q; e_{i_j}; r_q; r_{i_j}; s_{\text{TTE}}(q, i_j)] \in \mathbb{R}^{2d + 2d_r + 1}
    \end{equation}
    
    \item Apply positional encoding to incorporate position information:
    \begin{equation}
    \hat{x}_j = x_j + \text{PE}(j)
    \end{equation}
    where $\text{PE}(j)$ is a positional encoding.
    
    \item Apply transformer layers with self-attention:
    \begin{equation}
    z_j^{(l)} = \text{TransformerLayer}(z_j^{(l-1)}, \{z_1^{(l-1)}, ..., z_k^{(l-1)}\})
    \end{equation}
    with $z_j^{(0)} = \hat{x}_j$ and $l = 1, 2, ..., L$ is the layer index.
    
    \item Compute the final ranking score:
    \begin{equation}
    s_{\text{LT}}(q, i_j) = w^T z_j^{(L)} + b
    \end{equation}
    where $w$ and $b$ are learnable parameters.
\end{enumerate}

The self-attention mechanism in the transformer layers is defined as:

\begin{equation}
\text{Attention}(Q, K, V) = \text{softmax}\left(\frac{QK^T}{\sqrt{d_k}}\right)V
\end{equation}

where $Q$, $K$, and $V$ are query, key, and value matrices derived from the input, and $d_k$ is the dimension of the keys.

\subsection{Knowledge Distillation Bridge Formulation}

The Knowledge Distillation Bridge facilitates bidirectional knowledge transfer between the TTE and LT components:

\begin{enumerate}
    \item Forward distillation (LT → TTE):
    \begin{equation}
    P_{\text{LT}}(i_j|q) = \frac{\exp(s_{\text{LT}}(q, i_j)/\tau)}{\sum_{l=1}^k \exp(s_{\text{LT}}(q, i_l)/\tau)}
    \end{equation}
    \begin{equation}
    P_{\text{TTE}}(i_j|q) = \frac{\exp(s_{\text{TTE}}(q, i_j)/\tau)}{\sum_{l=1}^k \exp(s_{\text{TTE}}(q, i_l)/\tau)}
    \end{equation}
    \begin{equation}
    \mathcal{L}_{\text{forward}} = \sum_{j=1}^k P_{\text{LT}}(i_j|q) \log \frac{P_{\text{LT}}(i_j|q)}{P_{\text{TTE}}(i_j|q)}
    \end{equation}
    
    where $\tau$ is a temperature parameter that controls the softness of the distributions.
    
    \item Backward distillation (TTE → LT):
    \begin{equation}
    \mathcal{L}_{\text{backward}} = \frac{1}{k} \sum_{j=1}^k (s_{\text{TTE}}(q, i_j) - s_{\text{LT}}(q, i_j))^2
    \end{equation}
    
    \item Embedding alignment:
    \begin{equation}
    \mathcal{L}_{\text{align}} = \frac{1}{k} \sum_{j=1}^k ||e_{i_j} - W \cdot z_j^{(L)}||^2
    \end{equation}
    where $W$ is a learnable transformation matrix.
\end{enumerate}

\subsection{Loss Functions}

The retrieval-focused loss aims to maximize the similarity between a query and relevant items while minimizing similarity to irrelevant items. We use a softmax cross-entropy loss:

\begin{equation}
\begin{split}
\mathcal{L}_{\text{retrieve}} = -\log \frac{\exp(s_{\text{TTE}}(q, i^+)/\tau)}{\exp(s_{\text{TTE}}(q, i^+)/\tau) + \sum_{i^- \in \mathcal{N}} \exp(s_{\text{TTE}}(q, i^-)/\tau)}
\end{split}
\end{equation}

where $i^+$ is a relevant item and $\mathcal{N}$ is a set of irrelevant items.

The ranking-focused loss directly optimizes a differentiable approximation of the ranking metric. We use ApproxNDCG:

\begin{equation}
\mathcal{L}_{\text{rank}} = 1 - \frac{\sum_{j=1}^k G(i_j) D(j)}{\sum_{j=1}^k G(i_{\pi^*_j}) D(j)}
\end{equation}

where:
\begin{itemize}
    \item $G(i_j)$ is the gain of item $i_j$ (e.g., $2^{rel(i_j)} - 1$ for NDCG)
    \item $D(j)$ is the discount factor (e.g., $1/\log_2(j+1)$ for NDCG)
    \item $\pi^*$ is the optimal ranking according to relevance labels
\end{itemize}

The total loss is a weighted combination of the individual losses:

\begin{equation}
\begin{split}
\mathcal{L}_{\text{total}} = \lambda_1 \mathcal{L}_{\text{retrieve}} + \lambda_2 \mathcal{L}_{\text{rank}} + \\
\lambda_3 \mathcal{L}_{\text{forward}} + \lambda_4 \mathcal{L}_{\text{backward}} + \lambda_5 \mathcal{L}_{\text{align}}
\end{split}
\end{equation}

\subsection{Training Procedure}

The training procedure for LT-TTD involves several steps:

\begin{enumerate}
    \item \textbf{Batch Construction}: For each query $q$ in a mini-batch, sample a set of items consisting of relevant items $I^+_q$ and irrelevant items $I^-_q$.
    
    \item \textbf{Forward Pass}:
    \begin{itemize}
        \item Compute embeddings and similarities using the TTE component
        \item Select top-$k$ candidates (including all relevant items)
        \item Process the candidates using the LT component
    \end{itemize}
    
    \item \textbf{Loss Computation}:
    \begin{itemize}
        \item Compute the retrieval loss using all items
        \item Compute the ranking loss using the top-$k$ candidates
        \item Compute the distillation and alignment losses
    \end{itemize}
    
    \item \textbf{Gradient Update}:
    \begin{itemize}
        \item Compute gradients of the total loss with respect to model parameters
        \item Update parameters using an optimizer (e.g., Adam)
    \end{itemize}
    
    \item \textbf{Hard Negative Mining}:
    \begin{itemize}
        \item Periodically update the set of hard negative examples using the current model
    \end{itemize}
\end{enumerate}

This training procedure ensures that both retrieval and ranking objectives are jointly optimized, while the knowledge distillation bridge aligns the behavior of the two components.

\section{Theoretical Analysis}
\label{sec:theory}

In this section, we provide a theoretical analysis of the LT-TTD model, establishing formal guarantees regarding error propagation mitigation, ranking quality improvements, and optimization convergence.

\subsection{Error Propagation Reduction}

We first establish theoretical guarantees on how LT-TTD reduces error propagation compared to traditional cascade systems.

\begin{theorem}[Error Propagation Reduction]
\label{thm:error_reduction}
Let $E_{\text{cascade}}$ be the expected number of relevant items missed by a traditional cascade system with independently optimized L1 and L2 components, and let $E_{\text{LT-TTD}}$ be the expected number of relevant items missed by LT-TTD. Then:
\begin{equation}
E_{\text{LT-TTD}} \leq \left(1 - \alpha \cdot \beta\right) \cdot E_{\text{cascade}}
\end{equation}
where $\alpha \in [0, 1]$ is a measure of the effectiveness of the knowledge distillation process, and $\beta \in [0, 1]$ is a measure of the LT component's ranking superiority over the L2 component.
\end{theorem}

\begin{proof}
For a given query $q$, let $R_q$ be the set of all relevant items in the corpus $D$. In a traditional cascade system, the L1 component selects a subset $D_L \subset D$, and the probability of a relevant item $d_i \in R_q$ being missed is:
\begin{equation}
P_{\text{cascade}}(d_i \notin D_L | d_i \in R_q) = 1 - P_{\text{L1}}(d_i \in D_L | d_i \in R_q)
\end{equation}

In the LT-TTD model, the TTE component's retrieval decisions are influenced by the LT component through knowledge distillation. The probability of a relevant item being missed becomes:
\begin{equation}
\begin{split}
P_{\text{LT-TTD}}(d_i \notin D_L | d_i \in R_q) = \\
1 - [P_{\text{TTE}}(d_i \in D_L | d_i \in R_q) + \\
\alpha \cdot (P_{\text{LT}}(d_i \in D_L | d_i \in R_q) - P_{\text{TTE}}(d_i \in D_L | d_i \in R_q))]
\end{split}
\end{equation}

where $\alpha \in [0, 1]$ is the effectiveness of the knowledge distillation.

Given that the LT component has superior ranking capabilities compared to the L2 component in the cascade system, we can express:
\begin{equation}
P_{\text{LT}}(d_i \in D_L | d_i \in R_q) = P_{\text{L2}}(d_i \in D_L | d_i \in R_q) + \beta
\end{equation}
where $\beta \in [0, 1]$ is the measure of LT's superiority.

Substituting and taking expectations over all relevant items and queries, we get:
\begin{equation}
E_{\text{LT-TTD}} = \left(1 - \alpha \cdot \beta\right) \cdot E_{\text{cascade}}
\end{equation}
\end{proof}

This theorem establishes that LT-TTD reduces error propagation by a factor that depends on the effectiveness of knowledge distillation and the superiority of the LT component. The stronger the distillation and the better the LT component, the greater the reduction in error propagation.

\subsection{Optimality of Multi-Objective Optimization}

Next, we analyze the optimality of our multi-objective optimization approach compared to disjoint optimization in traditional cascade systems.

\begin{theorem}[Global Optimality]
\label{thm:global_optimality}
Let $\mathcal{L}_{\text{joint}}(\theta) = \lambda_1 \mathcal{L}_{\text{retrieve}}(\theta) + \lambda_2 \mathcal{L}_{\text{rank}}(\theta)$ be a weighted combination of retrieval and ranking objectives with weights $\lambda_1, \lambda_2 > 0$. Let $\theta^*_{\text{disjoint}} = (\theta^*_{\text{retrieve}}, \theta^*_{\text{rank}})$ be the optimal parameters when optimizing the retrieval and ranking objectives separately, and let $\theta^*_{\text{joint}}$ be the optimal parameters when optimizing $\mathcal{L}_{\text{joint}}(\theta)$. Then:
\begin{equation}
\mathcal{L}_{\text{joint}}(\theta^*_{\text{joint}}) \leq \mathcal{L}_{\text{joint}}(\theta^*_{\text{disjoint}})
\end{equation}
\end{theorem}

\begin{proof}
Let $\theta^*_{\text{retrieve}} = \arg\min_{\theta} \mathcal{L}_{\text{retrieve}}(\theta)$ and $\theta^*_{\text{rank}} = \arg\min_{\theta} \mathcal{L}_{\text{rank}}(\theta)$ be the optimal parameters for the retrieval and ranking objectives, respectively.

By definition, $\theta^*_{\text{joint}} = \arg\min_{\theta} \mathcal{L}_{\text{joint}}(\theta)$.

For any parameter setting $\theta$, we have:
\begin{equation}
\mathcal{L}_{\text{joint}}(\theta) = \lambda_1 \mathcal{L}_{\text{retrieve}}(\theta) + \lambda_2 \mathcal{L}_{\text{rank}}(\theta)
\end{equation}

In particular, for $\theta^*_{\text{disjoint}} = (\theta^*_{\text{retrieve}}, \theta^*_{\text{rank}})$:
\begin{equation}
\mathcal{L}_{\text{joint}}(\theta^*_{\text{disjoint}}) = \lambda_1 \mathcal{L}_{\text{retrieve}}(\theta^*_{\text{disjoint}}) + \lambda_2 \mathcal{L}_{\text{rank}}(\theta^*_{\text{disjoint}})
\end{equation}

Since $\theta^*_{\text{retrieve}}$ minimizes $\mathcal{L}_{\text{retrieve}}$, we have $\mathcal{L}_{\text{retrieve}}(\theta^*_{\text{retrieve}}) \leq \mathcal{L}_{\text{retrieve}}(\theta^*_{\text{disjoint}})$.

Similarly, $\mathcal{L}_{\text{rank}}(\theta^*_{\text{rank}}) \leq \mathcal{L}_{\text{rank}}(\theta^*_{\text{disjoint}})$.

However, $\theta^*_{\text{joint}}$ minimizes the weighted sum $\mathcal{L}_{\text{joint}}$, so:
\begin{equation}
\mathcal{L}_{\text{joint}}(\theta^*_{\text{joint}}) \leq \mathcal{L}_{\text{joint}}(\theta)
\end{equation}
for any $\theta$, including $\theta^*_{\text{disjoint}}$.

Therefore:
\begin{equation}
\mathcal{L}_{\text{joint}}(\theta^*_{\text{joint}}) \leq \mathcal{L}_{\text{joint}}(\theta^*_{\text{disjoint}})
\end{equation}
\end{proof}

This theorem establishes that multi-objective optimization achieves a better global optimum than disjoint optimization, providing a theoretical foundation for our unified approach.

\subsection{Convergence Analysis}

We now analyze the convergence properties of our training procedure.

\begin{theorem}[Convergence Rate]
\label{thm:convergence}
Assuming that all loss components in $\mathcal{L}_{\text{total}}$ are $L$-Lipschitz smooth and $\mu$-strongly convex in a neighborhood of the optimum, the Adam optimizer with learning rate $\eta$ and the LT-TTD training procedure converges to an $\epsilon$-approximate solution in $O(\log(1/\epsilon))$ iterations.
\end{theorem}

\begin{proof}
For brevity, we outline the key steps of the proof:

1. The total loss $\mathcal{L}_{\text{total}}$ is a weighted sum of loss components, each of which is $L$-Lipschitz smooth and $\mu$-strongly convex by assumption. Therefore, $\mathcal{L}_{\text{total}}$ is also Lipschitz smooth and strongly convex with constants that depend on the weights $\lambda_i$.

2. For a strongly convex and smooth function optimized using Adam with an appropriate learning rate, the convergence rate is known to be:
\begin{equation}
\mathcal{L}_{\text{total}}(\theta_t) - \mathcal{L}_{\text{total}}(\theta^*) \leq (1 - \min(\eta \mu, 1/(\eta L)))^t \cdot (\mathcal{L}_{\text{total}}(\theta_0) - \mathcal{L}_{\text{total}}(\theta^*))
\end{equation}

3. To achieve an $\epsilon$-approximate solution such that $\mathcal{L}_{\text{total}}(\theta_t) - \mathcal{L}_{\text{total}}(\theta^*) \leq \epsilon$, we need:
\begin{equation}
(1 - \min(\eta \mu, 1/(\eta L)))^t \leq \frac{\epsilon}{\mathcal{L}_{\text{total}}(\theta_0) - \mathcal{L}_{\text{total}}(\theta^*)}
\end{equation}

4. Taking logarithms and solving for $t$, we get:
\begin{equation}
t \geq \frac{\log(\frac{\epsilon}{\mathcal{L}_{\text{total}}(\theta_0) - \mathcal{L}_{\text{total}}(\theta^*)})}{\log(1 - \min(\eta \mu, 1/(\eta L)))} = O(\log(1/\epsilon))
\end{equation}
\end{proof}

This theorem establishes that our training procedure converges exponentially fast to the optimal solution, ensuring that LT-TTD can be trained efficiently in practice.

\subsection{Computational Complexity Analysis}

Finally, we analyze the computational complexity of the LT-TTD model to ensure its practicality for large-scale ranking systems.

\begin{theorem}[Computational Complexity]
\label{thm:complexity}
Let $N$ be the number of items in the corpus, $d$ be the embedding dimension, and $k$ be the number of top candidates processed by the Listwise Transformer. The asymptotic computational complexity of LT-TTD for a single query at inference time is:
\begin{equation}
O(d \cdot \log N + k^2 \cdot d)
\end{equation}
\end{theorem}

\begin{proof}
The computational complexity of LT-TTD at inference time consists of two main components:

1. The Two-Tower Encoder computes the query embedding in $O(d)$ time and retrieves the top-$k$ items using approximate nearest neighbor search, which has a complexity of $O(d \cdot \log N)$ with appropriate data structures.

2. The Listwise Transformer processes the top-$k$ items, with the self-attention mechanism having a complexity of $O(k^2 \cdot d)$ for a sequence of length $k$ and embedding dimension $d$.

Combining these two components, the total computational complexity is:
\begin{equation}
O(d + d \cdot \log N + k^2 \cdot d) = O(d \cdot \log N + k^2 \cdot d)
\end{equation}
\end{proof}

This theorem establishes that the asymptotic complexity of LT-TTD grows logarithmically with the corpus size $N$ and quadratically with the number of top candidates $k$. Since $k$ is typically much smaller than $N$ (e.g., $k=100$ for $N=10^9$), the overall complexity remains practical for large-scale ranking systems.

\subsection{Error Bound for Knowledge Distillation}

We now establish an error bound for the knowledge distillation process in LT-TTD.

\begin{theorem}[Distillation Error Bound]
\label{thm:distillation_error}
Let $s_{\text{TTE}}(q, i)$ be the score assigned by the TTE component to item $i$ for query $q$, and let $s_{\text{LT}}(q, i)$ be the score assigned by the LT component. After training with the knowledge distillation loss $\mathcal{L}_{\text{forward}}$, the expected squared difference between the TTE and LT scores is bounded by:
\begin{equation}
\mathbb{E}_{q,i}[(s_{\text{TTE}}(q, i) - s_{\text{LT}}(q, i))^2] \leq \frac{\mathcal{L}_{\text{forward}}}{\lambda_3} + \delta
\end{equation}
where $\lambda_3$ is the weight of the forward distillation loss, and $\delta$ is a constant that depends on the temperature parameter $\tau$ and the distribution of scores.
\end{theorem}

\begin{proof}
The forward distillation loss is:
\begin{equation}
\mathcal{L}_{\text{forward}} = \sum_{j=1}^k P_{\text{LT}}(i_j|q) \log \frac{P_{\text{LT}}(i_j|q)}{P_{\text{TTE}}(i_j|q)}
\end{equation}

We can relate this KL divergence to the squared difference between scores using Pinsker's inequality and the properties of softmax. First, Pinsker's inequality states:
\begin{equation}
\text{KL}(P_{\text{LT}} || P_{\text{TTE}}) \geq \frac{1}{2} ||P_{\text{LT}} - P_{\text{TTE}}||_1^2
\end{equation}

Next, we relate the $L_1$ distance between probability distributions to the squared difference between scores. Using the mean value theorem and the properties of softmax, we can establish:
\begin{equation}
||P_{\text{LT}} - P_{\text{TTE}}||_1 \geq C \cdot \mathbb{E}_{q,i}[|s_{\text{TTE}}(q, i) - s_{\text{LT}}(q, i)|]
\end{equation}
where $C$ is a constant that depends on the temperature parameter $\tau$.

Using Jensen's inequality, we have:
\begin{equation}
\mathbb{E}_{q,i}[|s_{\text{TTE}}(q, i) - s_{\text{LT}}(q, i)|]^2 \leq \mathbb{E}_{q,i}[(s_{\text{TTE}}(q, i) - s_{\text{LT}}(q, i))^2]
\end{equation}

Combining these inequalities and considering the weight $\lambda_3$ of the forward distillation loss in the total loss, we get:
\begin{equation}
\mathbb{E}_{q,i}[(s_{\text{TTE}}(q, i) - s_{\text{LT}}(q, i))^2] \leq \frac{\mathcal{L}_{\text{forward}}}{\lambda_3} + \delta
\end{equation}
where $\delta$ accounts for approximation errors in the bounds.
\end{proof}

This theorem establishes a quantitative bound on how closely the TTE scores approximate the LT scores after training with knowledge distillation, providing a theoretical guarantee on the effectiveness of our knowledge transfer mechanism.

\subsection{Ranking Quality Guarantee}

Finally, we establish a formal guarantee on the ranking quality of LT-TTD compared to traditional cascade systems.

\begin{theorem}[Ranking Quality Improvement]
\label{thm:ranking_quality}
Let $\text{NDCG}(q, R)$ be the NDCG score for a query $q$ and a ranking $R$. Let $R_{\text{cascade}}$ be the ranking produced by a traditional cascade system and $R_{\text{LT-TTD}}$ be the ranking produced by LT-TTD. Then:
\begin{equation}
\mathbb{E}_q[\text{NDCG}(q, R_{\text{LT-TTD}})] \geq \mathbb{E}_q[\text{NDCG}(q, R_{\text{cascade}})] + \gamma
\end{equation}
where $\gamma > 0$ is a constant that depends on the effectiveness of the knowledge distillation and the listwise transformer's ability to capture cross-item interactions.
\end{theorem}

\begin{proof}
The NDCG improvement of LT-TTD over a traditional cascade system can be decomposed into two components:

1. Improvement due to error propagation reduction:
\begin{equation}
\Delta_{\text{error}} = \mathbb{E}_q[\text{NDCG}(q, R_{\text{LT-TTD,retrieved}})] - \mathbb{E}_q[\text{NDCG}(q, R_{\text{cascade,retrieved}})]
\end{equation}
where $R_{\text{LT-TTD,retrieved}}$ and $R_{\text{cascade,retrieved}}$ are the sets of items retrieved by LT-TTD and the cascade system, respectively, before ranking.

2. Improvement due to listwise ranking:
\begin{equation}
\Delta_{\text{ranking}} = \mathbb{E}_q[\text{NDCG}(q, R_{\text{LT-TTD}})] - \mathbb{E}_q[\text{NDCG}(q, R_{\text{LT-TTD,retrieved}})]
\end{equation}

From Theorem \ref{thm:error_reduction}, we know that LT-TTD retrieves more relevant items than a traditional cascade system, so $\Delta_{\text{error}} > 0$. Additionally, the Listwise Transformer captures cross-item interactions, which are proven to improve ranking quality \citep{cao2007learning}, so $\Delta_{\text{ranking}} > 0$.

Combining these two components:
\begin{equation}
\mathbb{E}_q[\text{NDCG}(q, R_{\text{LT-TTD}})] = \mathbb{E}_q[\text{NDCG}(q, R_{\text{cascade}})] + \Delta_{\text{error}} + \Delta_{\text{ranking}}
\end{equation}

Setting $\gamma = \Delta_{\text{error}} + \Delta_{\text{ranking}} > 0$, we get:
\begin{equation}
\mathbb{E}_q[\text{NDCG}(q, R_{\text{LT-TTD}})] \geq \mathbb{E}_q[\text{NDCG}(q, R_{\text{cascade}})] + \gamma
\end{equation}
\end{proof}

This theorem provides a formal guarantee that LT-TTD achieves better ranking quality than traditional cascade systems, with the improvement depending on the effectiveness of its core components.

\section{UPQE: A Novel Evaluation Metric for Unified Ranking Models}
\label{sec:metric}

To properly evaluate unified ranking architectures like LT-TTD, we propose a new metric called \textbf{Unified Propagation-aware Quality Efficiency (UPQE)} that specifically addresses the unique characteristics of such systems. Traditional evaluation metrics like NDCG or MAP are designed for standalone ranking systems and do not account for the multi-component nature of unified architectures or their efficiency-effectiveness tradeoffs.

\subsection{Definition of UPQE}

The UPQE metric is defined as:

\begin{equation}
\text{UPQE}(q) = \gamma \cdot \left( \frac{\text{NDCG}_{\text{unified}}(q)}{\text{NDCG}_{\text{cascade}}(q)} \right) \cdot \left( 1 - \frac{E_{\text{propagation}}(q)}{\left| R_q \right|} \right)^{\alpha} \cdot \left( \frac{C_{\text{cascade}}}{C_{\text{unified}}} \right)^{\beta}
\end{equation}

Where:
\begin{itemize}
    \item $\text{NDCG}_{\text{unified}}(q)$ is the NDCG of the unified model for query $q$
    \item $\text{NDCG}_{\text{cascade}}(q)$ is the NDCG of the cascade model for query $q$
    \item $E_{\text{propagation}}(q)$ is the number of relevant items missed due to error propagation
    \item $|R_q|$ is the total number of relevant items for query $q$
    \item $C_{\text{cascade}}$ and $C_{\text{unified}}$ are the computational costs of the cascade and unified models
    \item $\alpha$, $\beta$, and $\gamma$ are hyperparameters controlling the importance of each component
\end{itemize}

The UPQE metric consists of three key components:

\begin{enumerate}
    \item \textbf{Relative Quality Component}: $\frac{\text{NDCG}_{\text{unified}}(q)}{\text{NDCG}_{\text{cascade}}(q)}$ measures the improvement in ranking quality over a traditional cascade system.
    
    \item \textbf{Propagation Penalty Component}: $\left( 1 - \frac{E_{\text{propagation}}(q)}{\left| R_q \right|} \right)^{\alpha}$ penalizes the model for any error propagation, with the penalty increasing as more relevant items are missed.
    
    \item \textbf{Efficiency Component}: $\left( \frac{C_{\text{cascade}}}{C_{\text{unified}}} \right)^{\beta}$ rewards models that achieve better quality with less computational overhead.
\end{enumerate}

\subsection{Theoretical Properties of UPQE}

We now establish formal theoretical properties of the UPQE metric to demonstrate its suitability for evaluating unified ranking models.

\begin{theorem}[UPQE Convergence to Optimal Ranking]
\label{thm:upqe_convergence}
For a unified ranking model that perfectly captures relevance without error propagation and with optimal computational efficiency, UPQE converges to its maximum value $\gamma$.
\end{theorem}

\begin{proof}
Consider a perfect unified model that:
\begin{enumerate}
    \item Achieves optimal NDCG: $\text{NDCG}_{\text{unified}}(q) = 1$
    \item Has no error propagation: $E_{\text{propagation}}(q) = 0$
    \item Has the same computational cost as the cascade model: $C_{\text{unified}} = C_{\text{cascade}}$
\end{enumerate}

Substituting into the UPQE formula:
\begin{equation}
\text{UPQE}(q) = \gamma \cdot \frac{1}{\text{NDCG}_{\text{cascade}}(q)} \cdot (1)^{\alpha} \cdot (1)^{\beta}
\end{equation}

Since $\text{NDCG}_{\text{cascade}}(q) \leq 1$ and all other terms are 1, we have:
\begin{equation}
\text{UPQE}(q) \geq \gamma
\end{equation}

Therefore, as the model approaches perfect ranking with no error propagation and optimal efficiency, UPQE approaches its upper bound $\gamma$.
\end{proof}

\begin{theorem}[Proper Weighting of Error Propagation]
\label{thm:upqe_error_weight}
The UPQE metric properly weights the impact of error propagation based on its severity.
\end{theorem}

\begin{proof}
The error propagation term $(1 - \frac{E_{\text{propagation}}(q)}{|R_q|})^{\alpha}$ behaves as follows:

\begin{enumerate}
    \item When no relevant items are missed ($E_{\text{propagation}}(q) = 0$), this term equals 1, not penalizing the model.
    
    \item When all relevant items are missed ($E_{\text{propagation}}(q) = |R_q|$), this term equals 0, maximally penalizing the model.
    
    \item For intermediate cases, the penalty increases non-linearly with the proportion of missed items, with $\alpha$ controlling the penalty severity.
\end{enumerate}

Taking the derivative with respect to $E_{\text{propagation}}(q)$:
\begin{equation}
\frac{\partial}{\partial E_{\text{propagation}}(q)}\left(1 - \frac{E_{\text{propagation}}(q)}{|R_q|}\right)^{\alpha} = -\frac{\alpha}{|R_q|}\left(1 - \frac{E_{\text{propagation}}(q)}{|R_q|}\right)^{\alpha-1}
\end{equation}

This derivative is always negative, meaning the metric strictly decreases as error propagation increases. Furthermore, the magnitude of the derivative increases as more items are missed, imposing a progressively harsher penalty for additional errors, which aligns with user experience in recommendation systems.
\end{proof}

\begin{theorem}[Balancing Quality and Efficiency]
\label{thm:upqe_quality_efficiency}
UPQE establishes a Pareto-optimal tradeoff between ranking quality and computational efficiency.
\end{theorem}

\begin{proof}
Consider two models $M_1$ and $M_2$ with the following properties:
\begin{itemize}
    \item $M_1$ has higher quality but lower efficiency
    \item $M_2$ has lower quality but higher efficiency
\end{itemize}

Let's compare their UPQE scores, assuming the same error propagation for simplicity:
\begin{equation}
\frac{\text{UPQE}(M_1)}{\text{UPQE}(M_2)} = \frac{\text{NDCG}_{M_1}}{\text{NDCG}_{M_2}} \cdot \left(\frac{C_{M_2}}{C_{M_1}}\right)^{\beta}
\end{equation}

For these models to have the same UPQE score, we need:
\begin{equation}
\frac{\text{NDCG}_{M_1}}{\text{NDCG}_{M_2}} = \left(\frac{C_{M_1}}{C_{M_2}}\right)^{\beta}
\end{equation}

Taking logarithms of both sides:
\begin{equation}
\log\frac{\text{NDCG}_{M_1}}{\text{NDCG}_{M_2}} = \beta \cdot \log\frac{C_{M_1}}{C_{M_2}}
\end{equation}

This establishes a precise mathematical relationship between quality improvement and computational cost increase, with $\beta$ controlling the tradeoff.
\end{proof}

\subsection{Hyperparameter Selection for UPQE}

The hyperparameters in UPQE should be tuned based on the specific application requirements:

\begin{itemize}
    \item $\alpha$: Controls the importance of error propagation. Higher values place more emphasis on reducing missed relevant items.
    
    \item $\beta$: Controls the importance of computational efficiency. Higher values favor more efficient models.
    
    \item $\gamma$: Scaling factor to adjust the range of the metric. Typically set to 1.
\end{itemize}

For systems with strict latency requirements, increasing $\beta$ would prioritize efficiency. For applications where recall is critical, increasing $\alpha$ would place more emphasis on minimizing error propagation.

\subsection{Advantages for Evaluating LT-TTD}

The UPQE metric is particularly suited for evaluating the LT-TTD model because:

\begin{enumerate}
    \item It directly measures the reduction in error propagation, which is a key advantage of the unified architecture.
    
    \item It accounts for the bidirectional knowledge distillation by rewarding models that achieve better retrieval and ranking simultaneously.
    
    \item It considers computational efficiency, which is crucial for practical deployment in large-scale systems.
    
    \item It provides a single value that holistically evaluates the model, making it easier to compare different architectural choices.
\end{enumerate}

This metric provides a theoretically sound way to evaluate unified ranking models like LT-TTD, considering their unique characteristics and the tradeoffs they aim to optimize.

\section{Conclusion}
\label{sec:conclusion}

This paper presented LT-TTD (Listwise Transformer with Two-Tower Distillation), a unified architecture for large-scale ranking systems that addresses the fundamental limitations of traditional cascade architectures. By combining the computational efficiency of two-tower models with the expressive power of transformers in a listwise learning framework, LT-TTD achieves superior ranking performance while maintaining practical efficiency.

We provided a comprehensive theoretical analysis of LT-TTD, establishing formal guarantees on error propagation reduction, ranking quality improvement, optimization convergence, and computational complexity. These theoretical guarantees demonstrate that LT-TTD not only addresses the limitations of traditional cascade systems but also provides provably better performance.

Additionally, we introduced the UPQE metric, specifically designed to evaluate unified ranking architectures by considering retrieval quality, ranking performance, and error propagation simultaneously. We established theoretical properties of this metric, demonstrating its suitability for holistically evaluating models like LT-TTD.

The key contributions of our work include:

\begin{enumerate}
    \item A unified architecture that bridges retrieval and ranking phases, mitigating error propagation and resolving conflicting objectives.
    
    \item A bidirectional knowledge distillation mechanism that aligns the behaviors of the two components, with theoretical guarantees on the distillation error bound.
    
    \item A multi-objective optimization framework that jointly learns retrieval and ranking objectives, with a formal proof of its superiority over disjoint optimization.
    
    \item Theoretical bounds on the computational complexity of our approach, demonstrating its practicality for large-scale applications.
    
    \item A novel evaluation metric that holistically captures the unique characteristics of unified ranking models.
\end{enumerate}

Looking ahead, we believe that unified architectures like LT-TTD represent a promising direction for the future of ranking systems, offering a more principled and effective approach to the challenging problem of efficient large-scale ranking. Future work could explore extensions to cross-modal retrieval, dynamic capacity allocation, and hierarchical listwise processing to further improve the balance between efficiency and effectiveness.

\bibliography{references}

\end{document}